\documentclass[12pt,a4paper]{article}
\usepackage{amssymb,amscd,amsmath,amsthm}
\usepackage[colorinlistoftodos]{todonotes}
\usepackage[colorlinks=true, allcolors=blue]{hyperref}
\usepackage{lmodern}
\usepackage[english]{babel}
\usepackage{color}
\usepackage{graphicx}
\usepackage{tikz}
 \usetikzlibrary{shapes,arrows,fit}

 \usepackage[lmargin=0.7in,rmargin=0.7in,tmargin=0.65in,bmargin=0.75in,footskip=0.25in]{geometry}

\usetikzlibrary{positioning, calc, arrows.meta}
\newtheorem{theorem}{Theorem}[section]

\newtheorem{corollary}[theorem]{Corollary}
\newtheorem{lemma}[theorem]{Lemma}

\newtheorem{remark}[theorem]{Remark}

\def\id{{\bf 1}\!\!{\rm I}}

\title{A Hidden Quantum Markov model framework for Entanglement and Topological Order in the AKLT Chain}

\date{}

\begin{document}

\maketitle

\centerline{ \author{\Large Abdessatar Souissi}}

\centerline{Department of Management Information Systems, College of Business and Economics, }
\centerline{Qassim University, Buraydah 51452, Saudi Arabia}
\centerline{\textit{a.souaissi@qu.edu.sa}}
\vskip0.3cm

\centerline{\author{\Large Amenallah Andolsi}}
\centerline{Nuclear Physics and High Energy Research Unit, Faculty of Sciences of Tunis,}
\centerline{ Tunis El Manar University, Tunis, Tunisia }
\centerline{\textit{amenallah.andolsi@fst.utm.tn }}
\vskip0.5cm

\begin{abstract}
This paper introduces a  hidden quantum Markov models (HQMMs) framework to the   Affleck–Kennedy–Lieb–Tasaki  (AKLT) state—a cornerstone example of a symmetry-protected topological (SPT) phase. The model's observation system is the physical spin-1 chain, which emerges from a hidden  spin-½ layer through well-defined quantum emission operation. We show that the underlying Markov dynamics caputure maximal entanglement through the use of significant channels relevant to the AKLT state. We also show that  SPT order induces a covariance on  the observation decoding channels.
This establishes an additional  bridge between the  quantum Machine learning and many-body physics, with promising implication in topological order and quantum information.
\end{abstract}


\section{Introduction}
Hidden Markov models \cite{Rab2002} are recognized as powerful stochastic processes due to their ability to capture complex correlations, with broad applications across various domains \cite{Mor2021}—especially in machine learning tasks \cite{Li17}, where algorithmic implementations  \cite{V03, FV05} have achieved significant success.
Several proposals have been made to generalize hidden Markov models into the quantum setting  \cite{Monras11,  Srin2017, MRDFS22}.  Namely, with the rising interest in quantum computing and quantum machine learning \cite{Li2024} but  a suitable generalized probabilistic picture is remained missing \cite{FLW24}. Recently, an algebraic approach to  HQMM within the framework of quantum probability has been suggested \cite{AGLS24Q}, offering promising features such as entanglement \cite{Abd23} and the ability to represent an important class of matrix product states as observations \cite{Sou25}. This approach defines a  HQMM as 
a hidden system—described by a Quantum Markov Chain (QMC) \cite{accardi1974, ASS20}, and an observation system strongly correlated with it. Information is encoded through sequences of hidden transitions and emission transition expectations, whose  dual maps are quantum channels \cite{AccOhy99}, these represent the fundamental building blocks of the HQMMs framework.
It is worth noting that QMCs serve as the quantum generalization of classical Markov chains and offer a powerful framework for describing the dynamics of quantum systems. They play a particularly foundational role in the theory of finitely correlated states on quantum spin chains \cite{FNW92, fannes2, fannes3}. Beyond this core application, QMCs find broad use in diverse areas such as quantum control~\cite{ticozzi2008}, quantum information theory~\cite{guan2024}, quantum software development~\cite{ying2016}, and quantum communication systems~\cite{wolf2012}. A particularly important subclass, quantum walks, has proven instrumental in developing novel quantum algorithms~\cite{ambainis2003}.

The AKLT model, introduced in \cite{AKLT1987, Affleck1988}, provides a paradigmatic example of a gapped spin-1 system with an exactly solvable ground state described by a Matrix Product State (MPS) \cite{FNW92}. It explicitly satisfies Haldane's conjecture \cite{Haldine83}, which distinguishes integer-spin antiferromagnetic chains—possessing a gapped spectrum (the “Haldane gap”) and exponentially decaying correlations—from their gapless, power-law-correlated half-integer counterparts. Beyond confirming  Haldane's conjecture, the AKLT state has emerged as a foundational model for SPT order \cite{Pollmann2010, Pollmann2012}. This stems from its elegant construction, where each physical spin-1 site arises from projecting two virtual spin-½ degrees of freedom into the symmetric triplet subspace, revealing a hidden layer of entanglement. At the virtual level, the system is constructed on auxiliary spin-½ spaces $\mathcal{H}_{\frac{1}{2}} \cong \mathbb{C}^2$, where the physical $SO(3)$ rotation symmetry is realized projectively through the double-cover map $SU(2) \to SO(3)$. Formally,  a  group $G$ of  symmetry acts on $\mathcal{H}_{\frac{1}{2}}$ via a projective unitary representation $\pi: G \to \mathcal{U}(\mathcal{H}_{\frac{1}{2}})$, characterized by a non-trivial 2-cocycle $\omega: G \times G \to U(1)$. Its cohomology class $[\omega] \in H^{2}(SO(3), U(1)) \cong \mathbb{Z}_2$ directly encodes the $\mathbb{Z}_2$ topological index of the SPT phase \cite{Chen2011b}. This algebraic structure arising from the mismatch between the virtual projective representation $\pi$ and the physical spin-1 representation $\rho: G \to \mathcal{U}(\mathcal{H}_1)$ acting on $\mathcal{H}_1 \cong \mathbb{C}^3$---underpins the robustness of fractionalized spin-½ edge modes and defines the SPT nature of the AKLT phase.

Previous works have shown a deep connection between quantum Markov chains on infinite tensor product of matrix algebras and finitely correlated states including matrix product states ansatz\cite{FNW92}. Moreover, recent algebraic approaches, framed within infinite tensor products of C*-algebras \cite{BR}, have been used to classify topological order in solid-state physics \cite{O21, O22, Tas23}.\\ 
In this context, this work establishes a  HQMM  framework for the AKLT spin chain, capturing its entanglement structure and  SPT order within a quantum probabilistic formalism. We construct the generative triplet $\Xi_{\mathrm{AKLT}} = (\phi_0, \mathcal{E}_H, \mathcal{E}_{O,H})$,
where $\phi_0: \mathcal{B}(\mathcal{H}_{\frac{1}{2}}) \to \mathbb{C}$ initializes the hidden spin-½ system, $\mathcal{E}_H: \mathcal{B}(\mathcal{H}_{\frac{1}{2}}) \otimes \mathcal{B}(\mathcal{H}_{\frac{1}{2}}) \to \mathcal{B}(\mathcal{H}_{\frac{1}{2}})$ governs its hidden dynamics, and $\mathcal{E}_{O,H}: \mathcal{B}(\mathcal{H}_{\frac{1}{2}}) \otimes \mathcal{B}(\mathcal{H}_{1}) \to \mathcal{B}(\mathcal{H}_{\frac{1}{2}})$ encodes the $G$-equivariant emission transition to the physical spin-$1$ space. The sequential quantum dynamics is defined by the transition map
\begin{equation*}
    E_{X_\ell,Y_\ell}(\cdot) = \mathcal{E}_H\big(\mathcal{E}_{O,H}(X_\ell \otimes Y_\ell) \otimes \cdot\big).
\end{equation*}

Within this framework, we demonstrate the following principal results:  Theorem~\ref{thm:main1} provides an explicit HQMM decomposition of the AKLT state, determining virtual-physical correlations within the model. We explicitly characterize the conditional dependency of physical observations on the underlying Markov process through the AKLT quantum channel.
 Theorem~\ref{thm:entropy}  studies the entanglement structure of the underlying Markov dynamics. It establishes that the Choi–Jamiołkowski state $J(\mathcal{E}_H^*)$ of the encoding channel $\mathcal{E}_H^*$ is a pure, maximally entangled bipartite state between the reference and output spaces. The resulting entanglement entropy $S(\rho_R) = S(\rho_B) = 1$ ebit confirms that the hidden dynamics fully encode the chain's entanglement structure \cite{FC2024, Kitaev06}. Theorem~\ref{thm:EOH-sym} explicitly formulates the $G$-covariance of the emission process, proving that $\mathcal{E}_{O,H}$ is $G$-equivariant. The covariance condition
\[
\mathcal{E}_{O,H}\big(\pi(g)X\pi(g)^\dagger \otimes \rho(g)Y\rho(g)^\dagger\big) = \pi(g)\,\mathcal{E}_{O,H}(X \otimes Y)\,\pi(g)^\dagger
\]
holds with respect to the projective virtual representation $\pi: G \to \mathcal{U}(\mathcal{H}_{\frac{1}{2}})$ and the physical representation $\rho: G \to \mathcal{U}(\mathcal{H}_1)$, thereby embedding the SPT index in the HQMM's symmetry structure.

This framework translates the topological and entanglement properties of the AKLT chain into a  well-defined quantum stochastic process. Our approach directly connects tensor network states \cite{SV12} to   quantum models \cite{BCLY20}, with promising implications for   quantum machine learning  and quantum computation architectures. Natural extensions include multi-dimensional generalizations to cluster states and projected entangled pair states \cite{Schuch2011, Chen2009}, with promising implications for measurement-based quantum computation \cite{Chen24, Von08} makes them prime candidates for this framework. The present work thus opens a pathway to investigating solid-state models through hidden quantum Markov dynamics, namely with the increasing efficiency of algebraic methods in understanding topological  order \cite{kap25, O22} .

We begin in Section \ref{Sect_prel} with the necessary background and definitions. Building on this foundation, Section \ref{Sect_HQMM_AKLT} is devoted to constructing the  HQMM for the AKLT chain. With this model established, we then explore the entanglement properties of its hidden dynamics in Section \ref{Sect_HiddenEntang}. Furthermore, in Section \ref{Sect_SPT_HQMM}, we reveal how this HQMM framework naturally encodes the Symmetry-Protected Topological (SPT) order characteristic of the AKLT state. We synthesize our findings and offer conclusions in Section \ref{Sect_Disc_Concl}.

\section{Preliminaries}\label{Sect_prel}
\subsection{ Hidden Quantum Markov Models}

Let $\mathcal{H}$ and $\mathcal{K}$ be two  finite-dimensional Hilbert spaces and let $\mathcal{B}(\mathcal{H})$ and $\mathcal{B}(\mathcal{K})$ be the algebras of bounded operators over the Hilbert spaces $\mathcal{H}$ and $\mathcal{K}$, respectively.  A map is $\mathcal{E}$ from  $\mathcal{B}(\mathcal{H})$ to  $\mathcal{B}(\mathcal{K})$ is said to be \emph{completely positive} if all its matrix amplifications
\[
\mathcal{E}: M_n(\mathcal{A}) \to M_n(\mathcal{B}), \quad [a_{ij}] \mapsto [\Phi(a_{ij})]
\]
preserve positivity for every $n \in \mathbb{N}$.   A \emph{quantum channel} is a completely positive, trace-preserving map $\Phi: \mathcal{A} \to \mathcal{B}$ that describes general physical transformations of quantum states. The fundamental duality between state and observable evolution is encoded in the dual map $\Phi^*: \mathcal{B}^* \to \mathcal{A}^*$, defined via the relation
\[
\operatorname{Tr}(\Phi(\rho) X) = \operatorname{Tr}(\rho \, \Phi^*(X)) \quad \text{for all } \rho \in \mathcal{A}_*, \, X \in \mathcal{B}.
\]
This establishes the precise correspondence between the Schrödinger picture (state evolution $\rho \mapsto \Phi(\rho)$) and the Heisenberg picture (observable evolution $X \mapsto \Phi^*(X)$).

A \emph{transition expectation} $\mathcal{E}: \mathcal{A} \otimes \mathcal{B} \to \mathcal{A}$ is a completely positive, identity-preserving map that generalizes the notion of conditional expectation to open quantum systems. The duality between quantum channels and transition expectations manifests through the adjoint relationship: if $\mathcal{E}$ is a transition expectation, then its dual $\mathcal{E}^*: \mathcal{A}^* \to (\mathcal{A} \otimes \mathcal{B})^*$ defines a quantum channel when properly normalized. This framework provides the mathematical foundation for describing emission, absorption, and general interaction processes in quantum information theory.

 HQMM  provides a generative description of quantum stochastic processes through a bipartite structure separating hidden dynamics from observable outputs. Formally, an HQMM is defined as a state on the quasi-local algebra $\mathcal{A}_{H,O} = (\mathcal{B}(\mathcal{H}) \otimes \mathcal{B}(\mathcal{H}_{1}))^{\otimes \mathbb{N}}$, representing an infinite bipartite quantum system where $\mathcal{A}_{H} = \mathcal{B}(\mathcal{H})^{\otimes \mathbb{N}}$ describes the hidden system and $\mathcal{A}_{O} = \mathcal{B}(\mathcal{H}_{1})^{\otimes \mathbb{N}}$ represents the observable output system  \cite{AGLS24Q}. The model is specified by a \emph{generative triplet} $\Xi = (\phi_0, \mathcal{E}_H, \mathcal{E}_{O,H})$, where:

\begin{itemize}
    \item $\phi_0: \mathcal{B}(\mathcal{H})\to \mathbb{C}$ is an initial state on the hidden algebra;
    \item $\mathcal{E}_{H}: \mathcal{B}(\mathcal{H}) \otimes \mathcal{B}(\mathcal{H}) \to \mathcal{B}(\mathcal{H})$ is a CPIP map governing hidden system dynamics;
    \item $\mathcal{E}_{O,H}: \mathcal{B}(\mathcal{H}) \otimes \mathcal{B}(\mathcal{H}_{1})\to \mathcal{B}(\mathcal{H})$ is a CPIP map encoding observable emission.
\end{itemize}

The conventional HQMM follows quantum causality where measurement precedes state evolution. For sequences $X = \bigotimes_{m=0}^n X_m$ and $Y = \bigotimes_{m=0}^n Y_m$, the state is:
\begin{equation}\label{eq:standard_HQMM}
\varphi_{H,O}(X \otimes Y) = \phi_0 \circ E_{X_0,Y_0} \circ E_{X_1,Y_1} \circ \cdots \circ E_{X_n,Y_n}(h),
\end{equation}
where the forward transition map is defined as:
\begin{equation}\label{eq:E_map}
E_{X,Y}(\cdot) := \mathcal{E}_H \big( \mathcal{E}_{O,H}(X \otimes Y) \otimes \cdot \big).
\end{equation}
This structure implements the causal sequence: at each step $m$, the emission map $\mathcal{E}_{O,H}$ first processes the observable $Y_m$ with the current hidden state, followed by hidden state propagation via $\mathcal{E}_{H}$.

\subsection{The AKLT State and its  SPT Order}\label{SbSect_AKLT}

The seminal work of Affleck, Kennedy, Lieb, and Tasaki \cite{Affleck1988} marked a watershed moment in quantum magnetism by furnishing an exactly solvable model that provided rigorous foundation for Haldane's conjecture regarding the existence of a spectral gap in antiferromagnetic spin-1 chains. The AKLT Hamiltonian, defined on a one-dimensional lattice with spin-1 degrees of freedom, embodies profound physical insights through its elegantly simple form:
\[
\hat{H}_{\text{AKLT}} = \sum_{\langle i,j \rangle} \left[ \vec{S}_i \cdot \vec{S}_j + \frac{1}{3} (\vec{S}_i \cdot \vec{S}_j)^2 \right],
\]
where \(\vec{S}_i = (S_i^x, S_i^y, S_i^z)\) are spin-1 operators satisfying \(\mathfrak{su}(2)\) commutation relations
\[
[S^x, S^y]  = iS^z, \quad [S^y, S^z]  = iS^x, \quad [S^z, S^x]  = iS^y
\]

The AKLT state is defined on a spin-1 chain, where the physical Hilbert space at each site is $\mathcal{H}_1 \cong \mathbb{C}^3$, spanned by the $S^z$-eigenbasis $\{|+\rangle, |0\rangle, |-\rangle\}$. Its construction relies on an auxiliary, virtual spin-½ space $\mathcal{H}_{\frac{1}{2}} \cong \mathbb{C}^2$ with the orthonormal basis $\{|\uparrow\rangle, |\downarrow\rangle\}$, which is projectively embedded into the physical spin-1 space.

The AKLT state admits a canonical MPS  representation that exactly captures its entanglement structure. The fundamental objects are three local matrices $\{A_+, A_0, A_-\} \in \mathcal{M}_2(\mathbb{C})$ that mediate between the physical spin-1 space and the auxiliary spin-½ space:

\begin{align}\label{eq:Ak}
A_{+} &= \sqrt{\frac{2}{3}} \, \sigma^{+}, &
A_{0} &= \sqrt{\frac{1}{3}} \, \sigma^{z}, &
A_{-} &= -\sqrt{\frac{2}{3}} \, \sigma^{-}
\end{align}

Here, the operators $\sigma^+$, $\sigma^-$, and $\sigma^z$ act on the virtual space $\mathbb{C}^2$, which is spanned by the basis vectors $|\uparrow\rangle = \begin{pmatrix} 1 \\ 0 \end{pmatrix}$ and $|\downarrow\rangle = \begin{pmatrix} 0 \\ 1 \end{pmatrix}$. Their explicit action is defined by the following outer products:
\begin{equation*}
\begin{aligned}
\sigma^+ &=   \begin{pmatrix} 0 & 1 \\ 0 & 0 \end{pmatrix}, \qquad
\sigma^- &=  \begin{pmatrix} 0 & 0 \\ 1 & 0 \end{pmatrix}, \qquad
\sigma^z &=   \begin{pmatrix} 1 & 0 \\ 0 & -1 \end{pmatrix}
\end{aligned}
\end{equation*}

For an $n$-site periodic chain, the AKLT state is constructed as the uniform MPS:

\begin{equation}\label{eqAKLT}
|\psi_{\mathrm{AKLT}}^{(n)}\rangle = \sum_{k_1,\dots,k_n \in \{+,0,-\}} \mathrm{Tr}\left(A_{k_1} A_{k_2} \cdots A_{k_n}\right) |k_1 k_2 \cdots k_n\rangle
\end{equation}

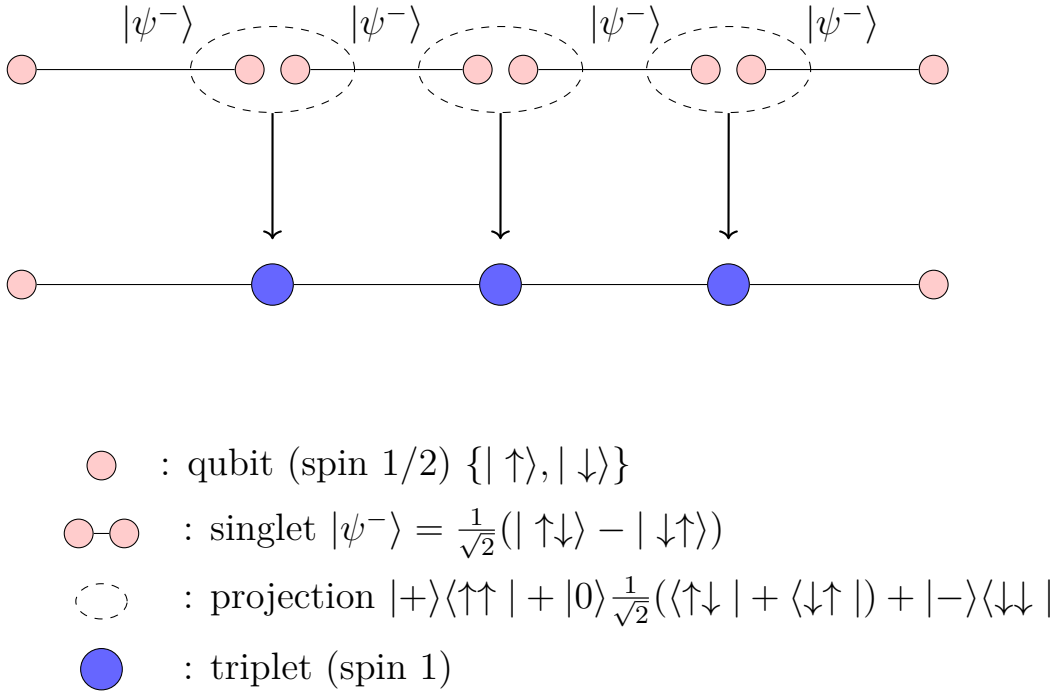
\begin{figure}[h!]
\centering
\begin{tikzpicture}[scale=1.5, every node/.style={scale=1.2}]
    \tikzstyle{qubit}=[circle, draw, fill=red!20, minimum size=9pt, inner sep=0pt]
    \tikzstyle{triplet}=[circle, draw, fill=blue!60, minimum size=13pt, inner sep=0pt]
    \tikzstyle{projection}=[ellipse, draw, dashed, minimum width=18pt, minimum height=12pt]

    \node[qubit] (q1) at (-0.7,3) {}; 
    \node[qubit] (q2) at (1.3,3) {};
    \node[qubit] (q3) at (1.7,3) {};
    \node[qubit] (q4) at (3.3,3) {};
    \node[qubit] (q5) at (3.7,3) {};
    \node[qubit] (q6) at (5.3,3) {};
    \node[qubit] (q7) at (5.7,3) {};
    \node[qubit] (q8) at (7.3,3) {};

    \draw (q1) -- (q2);
    \draw (q3) -- (q4);
    \draw (q5) -- (q6);
    \draw (q7) -- (q8);

    \node at (0.5,3.4) {$|\psi^-\rangle$};
    \node at (2.5,3.4) {$|\psi^-\rangle$};
    \node at (4.6,3.4) {$|\psi^-\rangle$};
    \node at (6.5,3.4) {$|\psi^-\rangle$};

    \node[projection, minimum width=10pt, minimum height=6pt, fit=(q2)(q3)] (p1) {};
    \node[projection, minimum width=14pt, minimum height=10pt, fit=(q4)(q5)] (p2) {};
    \node[projection, minimum width=14pt, minimum height=10pt, fit=(q6)(q7)] (p3) {};

    \draw[->, thick] (p1.south) -- (1.5,1.5);
    \draw[->, thick] (p2.south) -- (3.5,1.5);
    \draw[->, thick] (p3.south) -- (5.5,1.5);

    \node[triplet] (t1) at (1.5,1.1) {};
    \node[triplet] (t2) at (3.5,1.1) {};
    \node[triplet] (t3) at (5.5,1.1) {};

    \node[qubit] (qb1) at (-0.7,1.1) {};
    \node[qubit] (qb2) at (7.3,1.1) {};

    \draw (qb1) -- (t1);
    \draw (t1) -- (t2);
    \draw (t2) -- (t3);
    \draw (t3) -- (qb2);

    \node[qubit] (lq) at (0,-0.5) {};
    \node[anchor=west] at (0.4,-0.5) {: qubit (spin $1/2$) $\{|\uparrow\rangle,|\downarrow\rangle\}$};

    \node[qubit] (ls1) at (-0.2,-1.1) {};
    \node[qubit] (ls2) at (0.2,-1.1) {};
    \draw (ls1) -- (ls2);
    \node[anchor=west] at (0.6,-1.1) {: singlet $|\psi^-\rangle = \tfrac{1}{\sqrt{2}}(|\uparrow\downarrow\rangle - |\downarrow\uparrow\rangle)$};

    \node[projection, minimum width=16pt, minimum height=10pt] (lp) at (0,-1.7) {};
    \node[anchor=west] at (0.6,-1.7) {: projection $|+\rangle\langle\uparrow\uparrow| + |0\rangle\frac{1}{\sqrt{2}}(\langle \uparrow\downarrow| +\langle\downarrow\uparrow|) + |-\rangle\langle \downarrow\downarrow|$};

    \node[triplet] (lt) at (0,-2.3) {};
    \node[anchor=west] at (0.6,-2.3) {: triplet (spin 1)};
\end{tikzpicture}
\caption{Schematic diagram with eight qubits on the top row, grouped pairwise into singlet states (dashed ellipses), and projected downward into three spin-$1$ triplet nodes on the lower row. The increased vertical spacing highlights the correspondence between the two levels.}
\label{fig1}

\end{figure}

Figure \ref{fig1} illustrates how pairs of spin-½ particles can be combined to form an effective spin-$1$ degree of freedom. At the lower level, the red circles grouped in ellipses represent the tensor product space $\mathcal{H}_{\frac{1}{2}} \otimes \mathcal{H}_{\frac{1}{2}}$, which naturally decomposes into a three-dimensional triplet sector and a one-dimensional singlet sector. The upper blue circles correspond to the symmetric triplet space, which carries the spin-$1$ representation.

The  projection  $P: \mathcal{H}_{\frac{1}{2}} \otimes \mathcal{H}_{ \frac{1}{2}} \to \mathcal{H}_{1}$  implements the Clebsch-Gordan decomposition $\frac{1}{2} \otimes \frac{1}{2} = 1 \oplus 0$:
\begin{equation}
P = |+\rangle \langle \uparrow \uparrow| + |-\rangle \langle \downarrow \downarrow| + |0\rangle \frac{1}{\sqrt{2}}\left( \langle \uparrow \downarrow| + \langle \downarrow \uparrow| \right).
\end{equation}
This projection satisfies $P^\dagger P = \Pi_{\text{sym}}$ and $PP^\dagger = \id_{\mathcal{H}_{1}}$, establishing an isometric embedding that encodes the Haldane phase through its entanglement structure. Quantum channels channels are fundamental objects in the  analysis of the AKLT model. The AKLT channel is   defined  on $\mathcal{B}(\mathcal{H}_{\frac{1}{2}}) \equiv \mathbb{M}_2(\mathbb{C})$ as:
\begin{equation}\label{eq:aklE_channel}
    \Phi_{\mathrm{AKLT}}(Z) = \sum_{k \in \{+,0,-\}} A_k Z A_k^\dagger.
\end{equation}

The quantum channel $\Phi_{\mathrm{AKLT}}$ is bi-stochastic: it is both trace-preserving \(\operatorname{Tr}(\Phi_{\mathrm{AKLT}}(Z)) = \operatorname{Tr}(Z)\) and identity-preserving \(\Phi_{\mathrm{AKLT}}(\id_{\mathcal{H}_{1/2}}) = \id_{\mathcal{H}_{1/2}}\). When interpreted as a transfer matrix within the valence-bond-solid formalism, its spectral structure dictates the system's correlation functions: the largest eigenvalue is unity, while the second largest eigenvalue in magnitude, \(\lambda_2\), sets the dominant correlation length via \(\xi^{-1} = -\ln |\lambda_2|\) \cite{Pollmann2012}. The contractive nature of the channel upon iteration, a direct consequence of this spectral gap, underpins the mixing behavior and the exponential decay of correlations characteristic of the gapped phase of the AKLT state \cite{SA25}.

Let $G$ be a subgroup of the special orthogonal group in three dimensions $G = \mathrm{SO}(3)$, which constitutes our symmetry group. This group acts on the physical and virtual spaces of our system via distinct unitary representations. The physical Hilbert space is $\mathcal{H}_{1}$, on which $G$ acts via the fundamental, integer-spin representation $\rho: G \to \mathcal{U}(\mathcal{H}_{\frac{1}{2}})$, which is isomorphic to the spin-1 representation. The virtual, or auxiliary, Hilbert space is $\mathcal{H}_{\frac{1}{2}}$, where $G$ acts via the projective spin-½ representation $\pi: G \to \mathcal{U}(\mathcal{H}_{\frac{1}{2}})$. This representation $\pi$ is projective, meaning it satisfies the composition rule $\pi(g_1)\pi(g_2) = \omega(g_1,g_2)\pi(g_1g_2)$ for all $g_1, g_2 \in G$, with a non-trivial cocycle $\omega: G \times G \to \mathrm{U}(1)$. The cohomology class $[\omega]$ is a non-trivial element of the second group cohomology $H^2(\mathrm{SO}(3),\mathrm{U}(1))$, which is isomorphic to $\mathbb{Z}_2$, thereby encoding the double-cover relationship between $\mathrm{SO}(3)$ and $\mathrm{SU}(2)$ \cite{Chen2011b}.

A crucial property of the AKLT state, which guarantees its  non-trivial  SPT  phase  \cite{AKLT1987}, is the exact intertwining relation between the virtual representation $\pi$ and the physical representation $\rho$.  This relation is mathematically expressed as follows: For every $g \in G$ and for every  $k \in \{+, 0, -\}$, The MPS tensors $\{A_{k'}\}$ satisfy  the following covariance relation:
\begin{equation}\label{eq:Cov}
    \sum_{k'} \rho(g)_{k k'} A_{k'} = \pi(g) A_k \pi(g)^\dagger
\end{equation}
where $\rho(g)_{kk'} = \langle k | \rho(g) |k' \rangle$ denotes the matrix element of the operator $\rho(g)$. This equation means that a symmetry transformation $g$ applied to the physical index of the tensor is equivalent to a conjugation of the virtual operators by the projective representation $\pi(g)$ \cite{Pollmann2010}.

The symmetry action extends naturally to the dihedral group $D_2 \cong \mathbb{Z}_2 \times \mathbb{Z}_2$, a finite subgroup of $\mathrm{SO}(3)$ consisting of $\pi$-rotations about three orthogonal axes. This subgroup is a key symmetry of the AKLT Hamiltonian. The non-triviality of the SPT phase under this $D_2$ symmetry is detected by the Ogata index \cite{Tas23}, a $\mathbb{Z}_2$-valued topological invariant.
For a translation-invariant MPS that is invariant under the dihedral group $D_2 = \{e, g_x, g_y, g_z\}$ (the group of $\pi$-rotations about orthogonal axes), the SPT phase is classified by a $\mathbb{Z}_2$ index. This index is determined by the projective commutation relation of the symmetry actions on the virtual space.

Let $U_g$ be the unitary operators implementing the symmetry $g \in D_2$ on the physical spin-1 space.  For the generators $g_x$ and $g_y$, the projective algebra is:
\[
U_x^2 = \eta_x I, \quad U_y^2 = \eta_y I, \quad U_x U_y = \eta_{xy} U_y U_x,
\]
with $\eta_x, \eta_y, \eta_{xy} \in \{\pm 1\}$. The $\mathbb{Z}_2$ index characterizing the SPT phase is given by $\theta = \eta_{xy}$.

For the AKLT state, with virtual space dimension $d=2$, one finds $\eta_x = \eta_y = +1$ but $\eta_{xy} = -1$, yielding the non-trivial index $\theta = -1$. This value is a topological invariant, which can be computed via the gauge-invariant formula:
\[
\theta = \frac{1}{d} \mathrm{Tr}(U_x U_y U_x^\dagger U_y^\dagger).
\]
Substituting $d=2$ and $U_x U_y U_x^\dagger U_y^\dagger = -I$ for the AKLT state confirms $\theta = -1$ \cite{Pollmann2010, Pollmann2012}.

\section{HQMM Structure of the AKLT State}\label{Sect_HQMM_AKLT}
The AKLT state's unique structure can be naturally formulated as a  HQMM, revealing a profound connection between topological quantum matter and quantum stochastic processes. This framework cleanly separates the system into hidden virtual degrees of freedom that encode topological information and observable physical degrees accessible to measurement.

The \emph{Hidden Quantum Markov Model} for the AKLT chain is specified by the generative triplet $\Xi_{\mathrm{AKLT}} = (\phi_0, \mathcal{E}_H, \mathcal{E}_{O,H})$ where:
\begin{itemize}
    \item $\phi_0: \mathcal{B}(\mathcal{H}_{\frac{1}{2}}) \to \mathbb{C}$ initializes the hidden system
    \item $\mathcal{E}_H: \mathcal{B}(\mathcal{H}_{\frac{1}{2}}) \otimes \mathcal{B}(\mathcal{H}_{\frac{1}{2}}) \to \mathcal{B}(\mathcal{H}_{\frac{1}{2}})$ governs the hidden dynamics  of the spin-½ system.
    \item $\mathcal{E}_{O,H}: \mathcal{B}(\mathcal{H}_{\frac{1}{2}}) \otimes \mathcal{B}(\mathcal{H}_{1}) \to \mathcal{B}(\mathcal{H}_{\frac{1}{2}})$ emission transition expectation that encodes the spin-1 system through the spin-½ hidden system.
\end{itemize}
The sequential structure follows quantum causality through the transition map
\begin{equation}\label{eq:TXY}
  E_{X_\ell,Y_\ell}(\cdot) = \mathcal{E}_H(\mathcal{E}_{O,H}(X_\ell \otimes Y_\ell) \otimes \cdot)
\end{equation}
generating the complete state
\begin{equation}\label{eq_TXY}
 \varphi_{O,H}(X \otimes Y) = \phi_0 \circ E_{X_0,Y_0} \circ E_{X_1,Y_1} \circ \cdots \circ E_{X_n,Y_n}(\id_{\mathcal{H}_{\frac{1}{2}}}),
\end{equation}
where $X = X_1\otimes\cdots X_n\in \mathcal{B}(\mathcal{H}_{\frac{1}{2}})^{\otimes n}$ and $Y = Y_1\otimes\cdots Y_n\in \mathcal{B}(\mathcal{H}_{1})^{\otimes n}$.

This HQMM formulation provides a powerful framework for understanding how the AKLT state's topological properties emerge from the hidden quantum Markovian dynamics, which we will now analyze through the lens of symmetry protection.

 \begin{center}
     \begin{tikzpicture}[
    font=\small,
    alg/.style={circle, draw=black, thick, minimum size=1cm},
    map/.style={rectangle, draw=black, thick, rounded corners=3pt, fill=orange!20, minimum width=1.8cm, minimum height=0.8cm},
    arrow/.style={->, thick, >=latex},
    node distance=2.2cm and 1.2cm
]

 Title
\node[align=center, font=\large\bfseries] at (5,6)     {HQMM for AKLT Spin Chain};

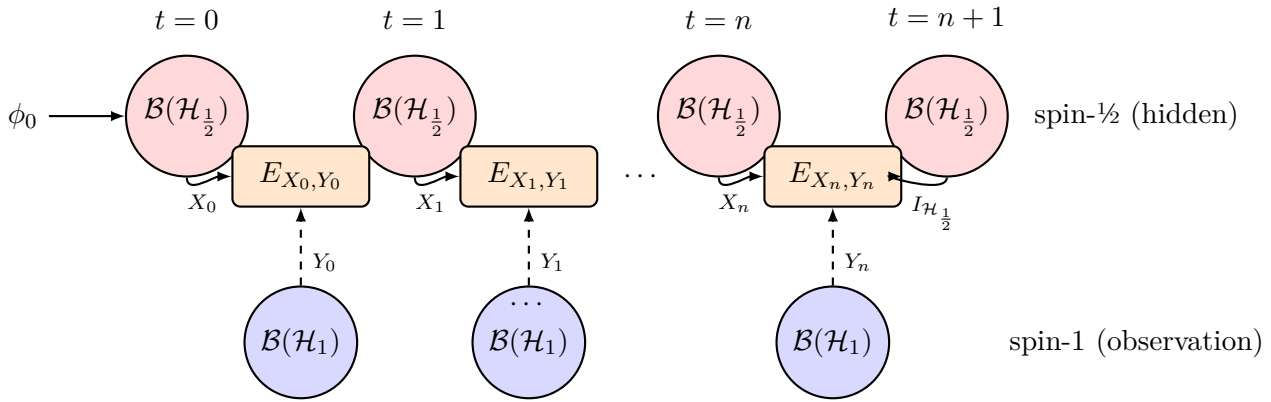
\node (hiddenlabel) at (12.5,4) { spin-½ (hidden)};
\node (obslabel) at (12.5,1) { spin-1 (observation) };

\node[alg, fill=red!15] (h0) at (0,4) {$\mathcal{B}(\mathcal{H}_{\frac{1}{2}})$};
\node[alg, fill=red!15] (h1) at (3,4) {$\mathcal{B}(\mathcal{H}_{\frac{1}{2}})$};
\node[alg, fill=red!15] (h2) at (7,4) {$\mathcal{B}(\mathcal{H}_{\frac{1}{2}})$};
\node at (6,3.2) {$\cdots$};
\node[alg, fill=red!15] (hn) at (10,4) {$\mathcal{B}(\mathcal{H}_{\frac{1}{2}})$};

\node[alg, fill=blue!15] (p0) at (1.5,1) {$\mathcal{B}(\mathcal{H}_{1})$};
\node[alg, fill=blue!15] (p1) at (4.5,1) {$\mathcal{B}(\mathcal{H}_{1})$};
\node[alg, fill=blue!15] (p2) at (8.5,1) {$\mathcal{B}(\mathcal{H}_{1})$};
\node at (4.5,1.5) {$\cdots$};

\node[map] (T0) at (1.5,3.2) {$E_{X_0,Y_0}$};
\node[map] (T1) at (4.5,3.2) {$E_{X_1,Y_1}$} ;
\node[map] (Tn) at (8.5,3.2) {$E_{X_n,Y_n}$};


\draw[arrow, dashed] (p0.north) -- node[right, pos=0.3] {\scriptsize $Y_0$} (T0.south);
\draw[arrow, dashed] (p1.north) -- node[right, pos=0.3] {\scriptsize $Y_1$} (T1.south);
\draw[arrow, dashed] (p2.north) -- node[right, pos=0.3] {\scriptsize $Y_n$} (Tn.south);

\draw[arrow] (h0.south) to[out=-90, in=180] node[below, pos=0.5] {\scriptsize $X_0$} (T0.west);
\draw[arrow] (h1.south) to[out=-90, in=180]  node[below, pos=0.5] {\scriptsize $X_1$} (T1.west);
\draw[arrow] (h2.south) to[out=-90, in=180] node[below, pos=0.5] {\scriptsize $X_n$} (Tn.west);
\draw[arrow] (hn.south) to[out= -90, in=180] node[below, pos=0.3] {\scriptsize $I_{\mathcal{H}_{\frac{1}{2}}}$} (Tn.east);

\node[left=1cm of h0] (phi0) {$\phi_0$};
\draw[arrow] (phi0.east) --   (h0.west);


\node[above=0.2cm of h0] {$t=0$};
\node[above=0.2cm of h1] {$t=1$};
\node[above=0.2cm of h2] {$t=n$};
\node[above=0.2cm of hn] {$t=n+1$};

\node[align=center, text width=14cm, font=\footnotesize] at (5,-1) {
    \textbf{Hidden correlations in the AKLT chain:} The HQMM represents sequential transitions between virtual spin-½ degrees of freedom (hidden system) and physical spin-1 measurements (observation
    system) via transition maps $E_{X_i,Y_i}$ that generate the complete quantum state correlations while       preserving the valence bond solid structure. \label{fig:aklE_hqmm}

};


\end{tikzpicture}
 \end{center}

  This HQMM formulation provides powerful quantum information tools to analyze entanglement scaling, quantum memory effects, and topological stability under decoherence, establishing fundamental connections between quantum many-body physics and quantum computational models.

In the subsequent analysis, we will explicitly derive the transition expectations governing the system's dynamics and elucidate the intricate correlation structure of the underlying hidden Markov model. This will be achieved by examining how the hidden algebra $\mathcal{B}(\mathcal{H}_{\frac{1}{2}})^{\otimes \mathbb{N}}$, which encodes the virtual spin network responsible for topological entanglement, correlates with the observation algebra $\mathcal{B}(\mathcal{H}_{1})^{\otimes \mathbb{N}}$, representing the physically measurable degrees of freedom, via the CPIP maps $\mathcal{E}_H$ and $\mathcal{E}_{O,H}$ implemented through precise virtual isometry $V$ that captures the hidden entanglement structure of the model and  an observation isometry $W$, that preserves the SPT order.

The mapping $V: \mathcal{H}_{\frac{1}{2}} \to \mathcal{H}_{\frac{1}{2}} \otimes \mathcal{H}_{\frac{1}{2}}$ is a partial isometry that builds entanglement between virtual spin spaces. A specific choice of $V$ is given by:
\begin{equation}\label{eq:V}
  V = |\Psi^-\rangle \langle \uparrow | + |\Psi^+\rangle \langle \downarrow |,
\end{equation}
where \(|\Psi^-\rangle = \frac{1}{\sqrt{2}}(|\uparrow\downarrow\rangle - |\downarrow\uparrow\rangle)\) is the singlet state forming the original AKLT valence bonds, and \(|\Psi^+\rangle = \frac{1}{\sqrt{2}}(|\uparrow\downarrow\rangle + |\downarrow\uparrow\rangle)\) is the symmetric triplet component. While this specific $V$ is a canonical choice, various other isometries can be selected to encode the system's entanglement and symmetry, each leading to a different physical state that can still saturate the maximal entanglement of the underlying virtual system.

Complementarily, the \emph{observation isometry} $W: \mathcal{H}_{\frac{1}{2}} \to \mathcal{H}_{\frac{1}{2}} \otimes \mathcal{H}_{1}$ maps virtual states to physical spin-1 observables:
\begin{equation}\label{eq:W}
  W\xi = \sum_{k \in \{+,0,-\}} A_{k}\xi \otimes |k\rangle,\quad \forall \xi \in \mathcal{H}_{\frac{1}{2}}
\end{equation}
where the matrices \(A_{k}\) are the standard AKLT MPS tensors. This structure explicitly reveals the AKLT state as a quantum autoregressive process, where sequential application of $W$ generates the physical spin-1 chain from the virtual space evolution.

\begin{lemma}\label{lem:Eh}
  The map $\mathcal{E}_H: \mathcal{B}(\mathcal{H}_{\frac{1}{2}}) \otimes \mathcal{B}(\mathcal{H}_{\frac{1}{2}}) \to \mathcal{B}(\mathcal{H}_{\frac{1}{2}})$ defined by the Stinespring dilation
\begin{equation}\label{eq:EH}
\mathcal{E}_H(X \otimes X') = V^\dagger(X \otimes X')V
\end{equation}
where $V: \mathcal{H}_{\frac{1}{2}} \to \mathcal{H}_{\frac{1}{2}} \otimes \mathcal{H}_{\frac{1}{2}}$ is the partial isometry from (\ref{eq:V}), constitutes a hidden transition expectation CPIP.
The dual quantum channel $\mathcal{E}_H^*: \mathcal{B}(\mathcal{H}_{\frac{1}{2}}) \to \mathcal{B}(\mathcal{H}_{\frac{1}{2}}) \otimes \mathcal{B}(\mathcal{H}_{\frac{1}{2}})$, defined via the Hilbert-Schmidt adjoint, acts as:
\begin{equation}\label{eq:EH_dual}
\mathcal{E}_H^*(\rho) = V\rho V^\dagger.
\end{equation}
This implements virtual bond preparation in the AKLT entanglement structure.
\end{lemma}
\begin{proof}
We first verify that $V$ is a partial isometry. From the definition (\ref{eq:V}), we compute $V^\dagger V$:
\begin{align*}
V^\dagger V |\uparrow\rangle &= V^\dagger\left( \frac{1}{\sqrt{2}}(|\uparrow\downarrow\rangle - |\downarrow\uparrow\rangle) \right)
= \frac{1}{2}(|\uparrow\rangle + |\downarrow\rangle) - \frac{1}{2}(-|\uparrow\rangle + |\downarrow\rangle) = |\uparrow\rangle, \\
V^\dagger V |\downarrow\rangle &= V^\dagger\left( \frac{1}{\sqrt{2}}(|\uparrow\downarrow\rangle + |\downarrow\uparrow\rangle) \right)
= \frac{1}{2}(|\uparrow\rangle + |\downarrow\rangle) + \frac{1}{2}(-|\uparrow\rangle + |\downarrow\rangle) = |\downarrow\rangle.
\end{align*}
Hence $V^\dagger V = \id$, proving $V$ is an isometry (and therefore a partial isometry).

Complete positivity of $\mathcal{E}_H$ follows from the Stinespring dilation: for any $n \in \mathbb{N}$ and $P \geq 0$ in $\mathcal{B}(\mathbb{C}^n \otimes \mathcal{H}_{\frac{1}{2}} \otimes \mathcal{H}_{\frac{1}{2}})$,
\[
(\id_n \otimes \mathcal{E}_H)(P) = (\id_n \otimes V)^\dagger P (\id_n \otimes V) \geq 0.
\]
Unitality is immediate:
\[
\mathcal{E}_H(\id_{\mathcal{H}_{\frac{1}{2}}} \otimes \id_{\mathcal{H}_{\frac{1}{2}}}) = V^\dagger(\id_{\mathcal{H}_{\frac{1}{2}}} \otimes \id_{\mathcal{H}_{\frac{1}{2}}})V = V^\dagger V = \id_{\mathcal{H}_{\frac{1}{2}}}.
\]

For the dual channel, take any $X,X',\rho \in \mathcal{B}(\mathcal{H}_{\frac{1}{2}})$:
\begin{align*}
\langle \mathcal{E}_H(X \otimes X'), \rho \rangle_{\mathrm{HS}}
&= \mathrm{Tr}[(V^\dagger(X \otimes X')V)^\dagger \rho] \\
&= \mathrm{Tr}[V^\dagger(X'^\dagger \otimes X^\dagger)V\rho] \\
&= \mathrm{Tr}[(X'^\dagger \otimes X^\dagger)V\rho V^\dagger] \\
&= \langle X \otimes X', V\rho V^\dagger \rangle_{\mathrm{HS}}.
\end{align*}
where we used the cyclic property of the trace. This proves $\mathcal{E}_H^*(\rho) = V\rho V^\dagger$.

For $\rho = p|\uparrow\rangle\langle\uparrow| + (1-p)|\downarrow\rangle\langle\downarrow|$, we compute:
\begin{align*}
\mathcal{E}_H^*(\rho) &= p \cdot V|\uparrow\rangle\langle\uparrow|V^\dagger + (1-p) \cdot V|\downarrow\rangle\langle\downarrow|V^\dagger \\
&= \frac{p}{2}(|\uparrow\downarrow\rangle - |\downarrow\uparrow\rangle)(\langle\uparrow\downarrow| - \langle\downarrow\uparrow|) \\
&\quad + \frac{1-p}{2}(|\uparrow\downarrow\rangle + |\downarrow\uparrow\rangle)(\langle\uparrow\downarrow| + \langle\downarrow\uparrow|) \\
&= \frac{1}{2}|\uparrow\downarrow\rangle\langle\uparrow\downarrow| + \frac{1}{2}|\downarrow\uparrow\rangle\langle\downarrow\uparrow| \\
&\quad + \frac{1-2p}{2}(|\uparrow\downarrow\rangle\langle\downarrow\uparrow| + |\downarrow\uparrow\rangle\langle\uparrow\downarrow|)
\end{align*}
This expression is verified by checking the extreme cases: when $p=1$ (pure $|\uparrow\rangle$), we recover the singlet state projector (up to normalization), and when $p=0$ (pure $|\downarrow\rangle$), we obtain the triplet state projector.
\end{proof}

\begin{remark}
  The hidden transition expectation $\mathcal{E}_H$ and its dual $\mathcal{E}_H^*$ link the dynamics of the AKLT chain at consecutive times. The dual channel $\mathcal{E}_H^*$ acts on density matrices, extending a state $\rho^{(n)}$ at time $n$ to an entangled bipartite state $\mathcal{E}_H^*(\rho^{(n)}) = V\rho^{(n)}V^\dagger$ at times $(n-1,n)$.

Conversely, $\mathcal{E}_H$ acts on operators. Its action on a generic operator $X$ is given by the conditional expectation:
\[
\mathcal{E}_H(X) = V^\dagger X V.
\]
This projects the joint operator $X$ on the bipartite system $(n-1,n)$ down to an operator on the marginal system at time $n$, enforcing the Markovian structure via $\rho^{(n)} = \mathcal{E}_H(\rho^{(n-1,n)})$.

This structure reveals the AKLT chain's quantum memory: the isometry $V$ creates the virtual singlets that propagate correlations, with the virtual bond space correlating successive temporal slices.

This dual pair of operations ($\mathcal{E}_H$, $\mathcal{E}_H^*$) constitutes a consistent family of conditional expectations and state preparations that implement the quantum Markov property. The composition $\mathcal{E}_H \circ \mathcal{E}_H^* = \id_{\mathcal{H}_{\frac{1}{2}}}$ ensures that extending then marginalizing returns the original state, while $\mathcal{E}_H^* \circ \mathcal{E}_H$ projects onto the range of $V$, which encodes the specific entanglement structure of the AKLT virtual bonds. This mathematical structure provides a precise formulation of how the HQMM captures the temporal correlations and entanglement propagation in symmetry-protected topological phases.
\end{remark}

\begin{lemma}\label{lem:emission_transition}
Let $W: \mathcal{H}_{\frac{1}{2}} \to \mathcal{H}_{\frac{1}{2}} \otimes \mathcal{H}_{1}$ be the partial isometry from (\ref{eq:W}). The  map $\mathcal{E}_{O,H}: \mathcal{B}(\mathcal{H}_{\frac{1}{2}}) \otimes \mathcal{B}(\mathcal{H}_{1}) \to \mathcal{B}(\mathcal{H}_{\frac{1}{2}})$ defined by
\begin{equation}\label{eq:EOH}
\mathcal{E}_{O,H}(X \otimes Y) = W^{\dagger} X\otimes Y W = \sum_{k,k'} \langle k' | Y | k \rangle  A_{k} X A_{k'}^{\dagger}
\end{equation}
is an emission transition expectation. Its dual channel $\mathcal{E}_{O,H}^*: \mathcal{B}(\mathcal{H}_{\frac{1}{2}}) \to \mathcal{B}(\mathcal{H}_{\frac{1}{2}}) \otimes \mathcal{B}(\mathcal{H}_{1})$ acts as:
\begin{equation}\label{eq:EOH_dual}
\mathcal{E}_{O,H}^*(Z) = \sum_{k,k'} A_{k}^{\dagger} Z A_{k'} \otimes |k\rangle\langle k'|.
\end{equation}
\end{lemma}

\begin{proof}
The proof proceeds by systematically establishing each property of the emission transition expectation and its dual channel. We begin by verifying that $W$ is indeed a partial isometry. Computing $W^\dagger W$ using the explicit form from (\ref{eq:W}) yields:
Here is the paraphrased text using \( A_k \) and \( A_{k'} \) notation:

\[
W^\dagger W = \left(\sum_k A_{k}^{\dagger} \otimes \langle k|\right)\left(\sum_{k'} A_{k'} \otimes |k'\rangle\right) = \sum_{k,k'} A_{k}^{\dagger} A_{k'} \otimes \langle k|k'\rangle = \sum_k A_k^{\dagger} A_{k} = \id_{\mathcal{H}_{\frac{1}{2}}},
\]
where the final equality represents the fundamental completeness relation satisfied by the AKLT MPS matrices, confirming the isometric property.

The equivalence between the Stinespring form and the explicit Kraus representation is established through direct computation:
\[
W^\dagger(X \otimes Y)W = \left(\sum_k A_{k}^{\dagger} \otimes \langle k|\right)(X \otimes Y)\left(\sum_{k'} A_{k'} \otimes |k'\rangle\right) = \sum_{k,k'} A_{k}^{\dagger} X A_{k'} \otimes \langle k|Y|k'\rangle = \sum_{k,k'} \langle k'|Y|k\rangle A_{k} X A_{k}^{\dagger}.
\]
where the adjoint symmetry of the inner product ensures the final form matches the definition.

Complete positivity follows immediately from the Stinespring dilation: for any positive operator $P$ in the amplified space, the map $(\id_n \otimes \mathcal{E}_{O,H})(P)$ manifests as a conjugation by the isometric extension $(\id_n \otimes W)$, preserving positivity. Unitality is verified through the chain:
\[
\mathcal{E}_{O,H}(\id_{\mathcal{H}_{\frac{1}{2}}} \otimes \id_{\mathcal{H}_{1}}) = W^\dagger(\id_{\mathcal{H}_{\frac{1}{2}}} \otimes \id_{\mathcal{H}_{1}})W = W^\dagger W = \id_{\mathcal{H}_{\frac{1}{2}}}
\]
demonstrating preservation of the identity.

The dual channel characterization emerges from the Hilbert-Schmidt adjoint relation. For arbitrary operators $X,Y,Z$, we systematically derive:
\[
\langle \mathcal{E}_{O,H}(X \otimes Y), Z \rangle_{\mathrm{HS}} = \mathrm{Tr}[(W^\dagger(X \otimes Y)W)^\dagger Z] = \langle X \otimes Y, W Z W^\dagger \rangle_{\mathrm{HS}},
\]
establishing $\mathcal{E}_{O,H}^*(Z) = W Z W^\dagger$. The explicit Kraus form follows by expanding this expression:
\begin{align*}
W Z W^\dagger &= \left(\sum_k A^{[k]} \otimes |k\rangle\right) Z \left(\sum_{k'} A^{[k']\dagger} \otimes \langle k'|\right) \\
&= \sum_{k,k'} A^{[k]} Z A^{[k']\dagger} \otimes |k\rangle\langle k'|,
\end{align*}
completing the demonstration of all stated properties.
\end{proof}

\begin{theorem}\label{thm:main1}
Let $\phi_0$ be an initial state on $\mathcal{B}(\mathcal{H}_{\frac{1}{2}})$, $\mathcal{E}_H: \mathcal{B}(\mathcal{H}_{\frac{1}{2}})\otimes \mathcal{B}(\mathcal{H}_{\frac{1}{2}})\to \mathcal{B}(\mathcal{H}_{\frac{1}{2}})$ the hidden transition expectation (\ref{eq:EH}), and $\mathcal{E}_{O,H}: \mathcal{B}(\mathcal{H}_{\frac{1}{2}})\otimes \mathcal{B}(\mathcal{H}_{1})\to \mathcal{B}(\mathcal{H}_{\frac{1}{2}})$ the emission transition expectation (\ref{eq:EOH}). For the HQMM $\varphi_{O,H}$ associated with the triple $\Xi = (\phi_0, \mathcal{E}_H, \mathcal{E}_{O,H})$ and local observables $X \otimes Y = \bigotimes_{i=1}^{n} X_i \otimes Y_i$, we have:
\begin{equation}\label{eq_phiOH}
\varphi_{O,H}(X \otimes Y) = \sum_{\substack{k_1,\dots,k_n \\ k'_1,\dots,k'_n}} \left( \prod_{\ell=1}^n \langle k_\ell | Y_\ell | k'_\ell \rangle \right) \varphi_{\substack{k_1,\dots,k_n \\ k'_1,\dots,k'_n}}(X),
\end{equation}
where the coefficient functionals are given by the nested composition:
\begin{equation}\label{eq:coefE_func}
\varphi_{\substack{k_1,\dots,k_n \\ k'_1,\dots,k'_n}}(X) = \phi_0 \left( V^{\dagger} \left( A_{k_1}X_1 A_{k'_1}^\dagger \otimes  \cdots \otimes V^{\dagger}\left( A_{k_n}X_n A_{k'_n}^\dagger \otimes \id_{\mathcal{H}} \right)V \cdots \right)V \right).
\end{equation}
\end{theorem}
\begin{proof}
 The fundamental building block (\ref{eq:E_map})  of the HQMMs   emerges from the  expression:
\[
E_{X_{\ell}, Y_{\ell}}(Z) = \mathcal{E}_{H}(\mathcal{E}_{O,H}(X_\ell \otimes Y_\ell) \otimes Z)
, \qquad \forall Z\in \mathcal{B}_H
\]
Substituting the emission transition expectation:
\begin{align*}
E_{X_n,Y_n}(\id_{\mathcal{H}}) &= \mathcal{E}_H \left( \mathcal{E}_{O,H}(X \otimes Y) \otimes \id_{\mathcal{H}_{\frac{1}{2}}} \right) \\
&\overset{(\ref{eq:EOH})}{=} \mathcal{E}_H \left( \sum_{k_n,k_n'} \langle k_n | Y_n | k_n' \rangle A_{k_n} X_n A_{k'_n}^\dagger \otimes \id_{\mathcal{H}_{\frac{1}{2}}} \right) \\
&\overset{(\ref{eq:EH})}{=} \sum_{k_n,k'_n} \langle k_n | Y_n | k'_n \rangle \mathcal{E}_H \left( A_{k_n} X_n A_{k'_n}^\dagger \otimes \id_{\mathcal{H}_{\frac{1}{2}}} \right)
\end{align*}

We now compute the $n$-fold composition by induction. For the base case $n=1$:
\[
\varphi_{O,H}(X_1 \otimes Y_1) = \phi_0 \circ E_{X_1,Y_1}(\mathcal{\id}_{\mathcal{H}}) = \sum_{k_1,k'_1} \langle k_1 | Y_1 | k'_1 \rangle \phi_0\left( \mathcal{E}_H(A_{k_1}X_1 A_{k'_1}^\dagger \otimes \id_{\mathcal{H}_{\frac{1}{2}}}) \right),
\]
which matches the theorem statement. For the inductive step, assume the decomposition holds for $n-1$:
\[
E_{X_2,Y_2} \circ \cdots \circ E_{X_n,Y_n}(\mathcal{\id}_{\mathcal{H}}) = \sum_{\substack{k_2,\dots,k_n \\ k'_2,\dots,k'_n}} \left( \prod_{\ell=2}^n \langle k_\ell | Y_\ell | k'_\ell \rangle \right) \mathcal{E}_H\left( A_{k_2}X_2 A_{k'_2}^\dagger \otimes \cdots \otimes \mathcal{E}_H\left( A_{k_n}X_n A_{k'_n}^\dagger \otimes \id_{\mathcal{H}_{\frac{1}{2}}} \right) \cdots \right).
\]
Applying $E_{X_1,Y_1}$ to this expression:
\begin{align*}
&E_{X_1,Y_1} \circ E_{X_2,Y_2} \circ \cdots \circ E_{X_n,Y_n}(\id_{\mathcal{H}_{\frac{1}{2}}}) \\
&= \sum_{k_1,k'_1} \langle k_1 | Y_1 | k'_1 \rangle \mathcal{E}_H \left( A_{k_1}X_1 A_{k'_1}^\dagger \otimes E_{X_2,Y_2} \circ \cdots \circ E_{X_n,Y_n}(\id_{\mathcal{H}_{\frac{1}{2}}}) \right) \\
&= \sum_{\substack{k_1,\dots,k_n \\ k'_1,\dots,k'_n}} \left( \prod_{\ell=1}^n \langle k_\ell | Y_\ell | k'_\ell \rangle \right) \mathcal{E}_H\left( A_{k_1}X_1 A_{k'_1}^\dagger \otimes \mathcal{E}_H\left( A_{k_2}X_2 A_{k'_2}^\dagger \otimes \cdots \otimes \mathcal{E}_H\left( A_{k_n}X_n A_{k'_n}^\dagger \otimes \id_{\mathcal{H}_{\frac{1}{2}}} \right) \cdots \right) \right).
\end{align*}
Applying the initial state $\phi_0$ yields the desired result:
\[
\varphi_{O, H}(X \otimes Y) = \sum_{\substack{k_1,\dots,k_n \\ k'_1,\dots,k'_n}} \left( \prod_{\ell=1}^n \langle k_\ell | Y_\ell | k'_\ell \rangle \right) \varphi_{\substack{k_1,\dots,k_n \\ k'_1,\dots,k'_n}}(X),
\]
where the coefficient functionals are precisely the nested compositions in \eqref{eq:coefE_func}. This completes the proof, demonstrating how the forward transition maps naturally generate the hierarchical structure of the HQMM.
\end{proof}

\begin{corollary}\label{cor:extended_hidden_markov}
In the notation of Theorem \ref{thm:main1}, the HQMM structure induces two naturally defined processes with distinct mathematical characteristics:

\begin{enumerate}
\item  The marginal process $\varphi_H$ on the hidden algebra $\mathcal{B}(\mathcal{H}_{\frac{1}{2}})^{\otimes \mathbb{N}}$, defined by
\[
\varphi_H(X) = \varphi_{O,H}(X \otimes \id_{\mathcal{H}_{1}}^{\otimes n}),
\]
constitutes a \emph{quantum Markov chain}. This process admits the explicit recursive representation:
\begin{equation}\label{eq:phiH}
\varphi_H(X) =  \sum_{ k_1,\dots,k_n} \phi_0 \left( V^{\dagger} \left( A_{k_1}X_1 A_{k'_1}^\dagger \otimes  \cdots \otimes V^{\dagger}\left( A_{k_n}X_n A_{k'_n}^\dagger \otimes \id_{\mathcal{H}} \right)V \cdots \right)V \right)
\end{equation}
where $\Phi_{\mathrm{AKLT}}: \mathcal{B}(\mathcal{H}_{\frac{1}{2}})\to \mathcal{B}(\mathcal{H}_{\frac{1}{2}})$ denotes the quantum channel associated with the AKLT model, as specified in (\ref{eq:aklE_channel}).

\item  The marginal process $\varphi_O$ on the observable algebra $\mathcal{B}(\mathcal{H}_{1})^{\otimes \mathbb{N}}$, obtained via
\begin{equation*}
\varphi_O(Y) = \varphi_{O,H}(\id_{\mathcal{H}_{\frac{1}{2}}}^{\otimes n} \otimes Y),
\end{equation*}
exhibits \emph{non-Markovian} behavior. Its structure is given by:
\begin{equation}\label{eq:phiO}
\varphi_O(Y) = \sum_{\substack{k_1,\dots,k_n \\ k'_1,\dots,k'_n}} \left( \prod_{\ell=1}^n \langle k_\ell | Y_\ell | k'_\ell \rangle \right) \varphi_{\substack{k_1,\dots,k_n \\ k'_1,\dots,k'_n}}(\id_{\mathcal{H}_{\frac{1}{2}}}^{\otimes n}),
\end{equation}
where the coefficient functionals $\varphi_{\substack{k_1,\dots,k_n \\ k'_1,\dots,k'_n}}$ are precisely those defined in equation (\ref{eq:coefE_func})
\end{enumerate}
\end{corollary}
\begin{proof}

\noindent\textbf{  1.} To establish that $\varphi_H$ is a quantum Markov chain with the representation \eqref{eq:phiH}, we begin with the definition:
\[
\varphi_H(X) = \varphi_{O,H}(X \otimes \id_{\mathcal{H}_{1}}^{\otimes n}).
\]
Substituting into Theorem~\ref{thm:main1} with $Y_\ell = \id_{\mathcal{H}_{1}}$ for all $\ell$, we observe that the matrix elements simplify as:
\[
\langle k_\ell | \id_{\mathcal{H}_{1}} | k'_\ell \rangle = \delta_{k_\ell, k'_\ell}.
\]
This diagonalization collapses the double summation in \eqref{eq_phiOH} to a single sum over indices $k_1, \dots, k_n$:
\[
\varphi_H(X) = \sum_{k_1,\dots,k_n} \varphi_{k_1,\dots,k_n \\ k_1,\dots,k_n}(X).
\]
Now examining the coefficient functionals \eqref{eq:coefE_func} under this diagonal constraint, we find that at each time step $\ell$, the expression $A_{k_\ell}X_\ell A_{k_\ell}^\dagger$ appears. Crucially, summing over $k_\ell$ yields the AKLT channel:
\[
\sum_{k_\ell} A_{k_\ell}X_\ell A_{k_\ell}^\dagger = \Phi_{\mathrm{AKLT}}(X_\ell).
\]
Substituting this into the nested structure of \eqref{eq:coefE_func} recursively yields the claimed representation \eqref{eq:phiH}.
The quantum Markov property is manifest in this recursive composition structure. Indeed, we can write:
\[
\varphi_H(X_1 \otimes \cdots \otimes X_n) = \phi_0 \circ E_{X_1} \circ E_{X_2} \circ \cdots \circ E_{X_n}(\id_{\mathcal{H}_{\frac{1}{2}}})
\]
where the transition maps $E_X: \mathcal{B}(\mathcal{H}_{\frac{1}{2}})\to \mathcal{B}(\mathcal{H}_{\frac{1}{2}})$ are defined by:
\[
E_X(\rho) = \mathcal{E}_H\left( \Phi_{\mathrm{AKLT}}(X) \otimes \rho \right).
\]
This composition structure explicitly demonstrates that the process depends only on the current state, satisfying the quantum Markov property.

\noindent\textbf{1.} For the observation process, we set $X = \id_{\mathcal{H}_{\frac{1}{2}}}^{\otimes n}$ in Theorem~\ref{thm:main1}, yielding immediately:
\[
\varphi_O(Y) = \varphi_{O,H}(\id_{\mathcal{H}_{\frac{1}{2}}}^{\otimes n} \otimes Y) = \sum_{\substack{k_1,\dots,k_n \\ k'_1,\dots,k'_n}} \left( \prod_{\ell=1}^n \langle k_\ell | Y_\ell | k'_\ell \rangle \right) \varphi_{\substack{k_1,\dots,k_n \\ k'_1,\dots,k'_n}}(\id_{\mathcal{H}_{\frac{1}{2}}}^{\otimes n}),
\]
which is exactly \eqref{eq:phiO}. The non-Markovian character of $\varphi_O$ follows from the fact that the coefficient functionals $\varphi_{\substack{k_1,\dots,k_n \\ k'_1,\dots,k'_n}}(\id_{\mathcal{H}_{\frac{1}{2}}}^{\otimes n})$ encode the full history of the hidden process through the nested composition in \eqref{eq:coefE_func}. Unlike the Markovian hidden process, the observation statistics at any time depend on the entire past history of hidden state evolution, mediated through the emission transition expectation $\mathcal{E}_{O,H}$.
\end{proof}

The hidden process \(\varphi_H\) (Eq. \ref{eq:phiH}) is a  QMC  on  the infinite tensor product algebera $\mathcal{B}(\mathcal{H}_{\frac{1}{2}})^{\otimes \mathbb{N}}$, which is  strongly  related to the AKLT structure. Namely, its correlations are governed by the AKLT quantum channel. This means that its Markovian evolution directly reflects the physical properties of the AKLT spin chain.

This provides a concrete, physical example of a QMC, connecting an abstract stochastic process to a well-studied many-body system. This link offers a powerful new framework that allows more significant probabilistic interpretations for SPT phases.

\section{Entanglement Dynamics in Hidden Markov Channel}\label{Sect_HiddenEntang}

To analyze the entanglement properties of a quantum channel $\Lambda: \mathcal{B}(\mathcal{H}_{\text{in}}) \to \mathcal{B}(\mathcal{H}_{\text{out}})$, we employ the Choi–Jamiołkowski isomorphism:
\[
J(\Lambda) = (\mathrm{id}_{\text{in}} \otimes \Lambda)\bigl(|\Omega\rangle\langle\Omega|\bigr),
\]
where $|\Omega\rangle = \frac{1}{\sqrt{d_{\text{in}}}}\sum_{i=1}^{d_{\text{in}}} |i\rangle_{\text{ref}} \otimes |i\rangle_{\text{in}}$ is a maximally entangled state between a reference system $\mathcal{H}_{\text{ref}} \cong \mathcal{H}_{\text{in}}$ and the input system $\mathcal{H}_{\text{in}}$ \cite{FC2024}. This construction provides a complete characterization of the channel's ability to preserve or generate entanglement.
  Our primary system is the spin-½ Hilbert space $\mathcal{H}_{\frac12} \cong \mathbb{C}^2$, with the standard basis $\{|\uparrow\rangle, |\downarrow\rangle\}$ representing spin-up and spin-down states. We define two particular entangled states of two qubits:
\[
|\Psi^-\rangle = \frac{1}{\sqrt{2}}(|\uparrow\downarrow\rangle - |\downarrow\uparrow\rangle), \qquad
|\Psi^+\rangle = \frac{1}{\sqrt{2}}(|\uparrow\downarrow\rangle + |\downarrow\uparrow\rangle),
\]
which are two of the four Bell states. The two-dimensional subspace spanned by these states is denoted by $\mathcal{S} = \operatorname{span}\{|\Psi^-\rangle, |\Psi^+\rangle\} \subset \mathcal{H}_{\frac12} \otimes \mathcal{H}_{\frac12}$.

The isometry $V: \mathcal{H}_{\frac12} \to \mathcal{H}_{\frac12} \otimes \mathcal{H}_{\frac12}$ encodes a single spin into a specific two-spin entangled subspace:
\[
V = |\Psi^-\rangle\langle\uparrow| + |\Psi^+\rangle\langle\downarrow|.
\]
This map satisfies $V^\dagger V = I_2$, confirming it preserves the inner product, and its image is precisely $\mathcal{S}$.

From this isometry, we naturally obtain two complementary quantum channels. The \emph{encoding channel} $\mathcal{E}_H^*: \mathcal{B}(\mathcal{H}_{\frac12}) \to \mathcal{B}(\mathcal{H}_{\frac12} \otimes \mathcal{H}_{\frac12})$ embeds a single-qubit state into the two-qubit space:
\[
\mathcal{E}_H^*(\rho) = V\rho V^\dagger.
\]
The adjoint channel $\mathcal{E}_H: \mathcal{B}(\mathcal{H}_{\frac12} \otimes \mathcal{H}_{\frac12}) \to \mathcal{B}(\mathcal{H}_{\frac12})$ acts as a decoder, extracting information from the two-qubit system back to a single qubit:
\[
\mathcal{E}_H(X) = V^\dagger X V.
\]

\begin{theorem}\label{thm:entropy}
Consider the encoding channel $\mathcal{E}_H^*$ defined above. Its Choi–Jamiołkowski state $J(\mathcal{E}_H^*)$ is a pure state exhibiting maximal entanglement between the reference system and the output subsystem. Concretely,
\[
J(\mathcal{E}_H^*) = |\Psi_{\mathrm{Choi}}\rangle\langle\Psi_{\mathrm{Choi}}|,
\quad
|\Psi_{\mathrm{Choi}}\rangle = \frac{1}{\sqrt{2}}\bigl( |\uparrow\rangle_R \otimes |\Psi^-\rangle_B + |\downarrow\rangle_R \otimes |\Psi^+\rangle_B \bigr),
\]
where $R$ labels the reference qubit and $B$ the output two-qubit system. The entanglement entropy across the $R|B$ partition is
\[
S\bigl(\rho_R\bigr) = S\bigl(\rho_B\bigr) = 1 \; \text{ebit},
\]
with $\rho_R = \operatorname{Tr}_B J(\mathcal{E}_H^*)$ and $\rho_B = \operatorname{Tr}_R J(\mathcal{E}_H^*)$. This shows that $\mathcal{E}_H^*$ perfectly preserves one ebit of entanglement with a reference system, achieving the theoretical maximum for an encoding from one qubit to two qubits restricted to the Bell subspace $\mathcal{S}$.
\end{theorem}

\begin{proof}
We proceed through a systematic construction and analysis of the Choi state.
Since $\mathcal{E}_H^*$ takes a single qubit as input, we have $d_{\text{in}}=2$. We identify the reference system basis with the input basis: $|\uparrow\rangle_R$ corresponds to $|\uparrow\rangle_{\text{in}}$ and $|\downarrow\rangle_R$ to $|\downarrow\rangle_{\text{in}}$. The maximally entangled state between reference and input is therefore
\[
|\Omega\rangle = \frac{1}{\sqrt{2}}\bigl(|\uparrow\rangle_R|\uparrow\rangle_{\text{in}} + |\downarrow\rangle_R|\downarrow\rangle_{\text{in}}\bigr).
\]
Applying the identity on the reference and the channel on the input gives:
\[
J(\mathcal{E}_H^*) = \frac{1}{2}\sum_{i,j \in \{\uparrow,\downarrow\}} |i\rangle\langle j|_R \otimes V|i\rangle\langle j|_{\text{in}} V^\dagger.
\]
Using the defining action of $V$, namely $V|\uparrow\rangle = |\Psi^-\rangle$ and $V|\downarrow\rangle = |\Psi^+\rangle$, we obtain:
\[
J(\mathcal{E}_H^*) = \frac{1}{2}\Bigl( |\uparrow\rangle\langle\uparrow|_R \otimes |\Psi^-\rangle\langle\Psi^-|_B + |\uparrow\rangle\langle\downarrow|_R \otimes |\Psi^-\rangle\langle\Psi^+|_B + |\downarrow\rangle\langle\uparrow|_R \otimes |\Psi^+\rangle\langle\Psi^-|_B + |\downarrow\rangle\langle\downarrow|_R \otimes |\Psi^+\rangle\langle\Psi^+|_B \Bigr).
\]
This operator is precisely the projector onto the pure state $|\Psi_{\mathrm{Choi}}\rangle$ as stated in the theorem.

\vspace{0.5em}

Since $|\Psi_{\mathrm{Choi}}\rangle$ is pure, we can compute the reduced states by partial tracing. For the reference system:
\begin{align*}
\rho_R &= \operatorname{Tr}_B |\Psi_{\mathrm{Choi}}\rangle\langle\Psi_{\mathrm{Choi}}|\\
&= \frac{1}{2}\bigl(\langle\Psi^-|\Psi^-\rangle |\uparrow\rangle\langle\uparrow| + \langle\Psi^-|\Psi^+\rangle|\uparrow\rangle\langle\downarrow| + \langle\Psi^+|\Psi^-\rangle|\downarrow\rangle\langle\uparrow| + \langle\Psi^+|\Psi^+\rangle|\downarrow\rangle\langle\downarrow| \bigr)
\end{align*}

The Bell states are orthonormal: $\langle\Psi^\pm|\Psi^\pm\rangle=1$ and $\langle\Psi^-|\Psi^+\rangle=0$. Thus,
\[
\rho_R = \frac{1}{2}\bigl(|\uparrow\rangle\langle\uparrow| + |\downarrow\rangle\langle\downarrow|\bigr) = \frac{I_2}{2}.
\]
The von Neumann entropy of this completely mixed qubit state is $S(\rho_R) = -\operatorname{Tr}(\rho_R \log_2 \rho_R) = 1$ ebit.

For the output system $B$:
\[
\rho_B = \operatorname{Tr}_R |\Psi_{\mathrm{Choi}}\rangle\langle\Psi_{\mathrm{Choi}}|
= \frac{1}{2}\bigl( |\Psi^-\rangle\langle\Psi^-| + |\Psi^+\rangle\langle\Psi^+| \bigr).
\]
To express this in the computational basis, we write:
\[
|\Psi^-\rangle = \frac{|\uparrow\downarrow\rangle-|\downarrow\uparrow\rangle}{\sqrt{2}}, \qquad
|\Psi^+\rangle = \frac{|\uparrow\downarrow\rangle+|\downarrow\uparrow\rangle}{\sqrt{2}}.
\]
Then,
\[
|\Psi^-\rangle\langle\Psi^-| = \frac{1}{2}\bigl(|\uparrow\downarrow\rangle\langle\uparrow\downarrow| - |\uparrow\downarrow\rangle\langle\downarrow\uparrow| - |\downarrow\uparrow\rangle\langle\uparrow\downarrow| + |\downarrow\uparrow\rangle\langle\downarrow\uparrow|\bigr),
\]
\[
|\Psi^+\rangle\langle\Psi^+| = \frac{1}{2}\bigl(|\uparrow\downarrow\rangle\langle\uparrow\downarrow| + |\uparrow\downarrow\rangle\langle\downarrow\uparrow| + |\downarrow\uparrow\rangle\langle\uparrow\downarrow| + |\downarrow\uparrow\rangle\langle\downarrow\uparrow|\bigr).
\]
Adding these and dividing by 2 yields:
\[
\rho_B = \frac{1}{2}\bigl( |\uparrow\downarrow\rangle\langle\uparrow\downarrow| + |\downarrow\uparrow\rangle\langle\downarrow\uparrow| \bigr).
\]
This density matrix has eigenvalues $\{\tfrac12, \tfrac12, 0, 0\}$, giving $S(\rho_B) = 1$ ebit as well.

For a pure bipartite state, the entanglement entropy is given by the von Neumann entropy of either reduced density matrix. Here, $S(\rho_R) = S(\rho_B) = 1$. Since the reference system is a single qubit, the maximum possible entanglement entropy is $\log_2 2 = 1$ ebit. Our calculated value saturates this bound, confirming that $|\Psi_{\mathrm{Choi}}\rangle$ is maximally entangled across the $R|B$ partition given the dimension constraints.
\end{proof}

 \begin{remark}
The encoding channel $\mathcal{E}_H^*$ maps a single qubit isometrically into the two-dimensional Bell subspace $\mathcal{S}$ of two qubits. The fact that its Choi state carries exactly one ebit of entanglement means that the channel can perfectly transmit one qubit of quantum information—its coherent information is maximal. Moreover, it preserves the maximum possible entanglement with a reference system, which is the hallmark of an optimal quantum encoding. The corresponding decoding channel $\mathcal{E}_H$ reverses this process on the subspace $\mathcal{S}$, and both channels share the same entanglement capacity as quantified by their Choi states. This structure reflects the essential role of such encodings in valence-bond solid models and quantum error correction schemes.
\end{remark}
\section{SPT order on the AKLT,s HQMMs}\label{Sect_SPT_HQMM}
The AKLT state's remarkable robustness stems from its symmetry-protected topological (SPT) order, which manifests elegantly within its Hidden Quantum Markov Model (HQMM) description. To understand this protection mechanism,  clarify how symmetries operate in the HQMM framework.

Let $G$ be a symmetry group with unitary representations $\pi: G \to \mathcal{U}(\mathcal{H}_{1/2})$ on the virtual space and $\rho: G \to \mathcal{U}(\mathcal{H}_1)$ on the physical space.

The heart of the matter lies in the \emph{emission transition expectation} $\mathcal{E}_{O,H}$, which implements the quantum-classical interface between hidden and observable sectors. This channel is built from the AKLT tensors $A_k$ through:
\begin{equation}\label{eq:EOH2}
  \mathcal{E}_{O,H}(Z \otimes Y) = \sum_{k,k'} \langle k'|Y|k\rangle A_k Z A_{k'}^\dagger
\end{equation}
The fundamental covariance condition \cite{AKLT1987,Pollmann2010}:
\begin{equation*}\label{eq:Cov}
    \sum_{k'} \rho(g)_{k k'} A_{k'} = \pi(g) A_k \pi(g)^\dagger
\end{equation*}
ensures that $\mathcal{E}_{O,H}$ transforms covariantly, meaning that symmetry transformations commute with the emission process. This covariance propagates to the composite transition map $E_{X,Y}$ that drives the HQMM dynamics:
\[
E_{X,Y}(Z) = \mathcal{E}_H(\mathcal{E}_{O,H}(X \otimes Y) \otimes Z)
\]
The map $E_{X,Y}$ inherits the symmetry property:
\[
E_{\rho(g)X\rho(g)^\dagger, \rho(g)Y\rho(g)^\dagger}(\pi(g)Z\pi(g)^\dagger) = \pi(g)E_{X,Y}(Z)\pi(g)^\dagger
\]

This hierarchical symmetry preservation---from individual channels to composite dynamics---ensures that the AKLT state's topological character remains protected throughout the quantum Markov process. The HQMM framework thus reveals how SPT order emerges from the consistent interplay between hidden dynamics and symmetry constraints across all scales of the system.

 \begin{theorem}\label{thm:EOH-sym}
Let $G$ be a symmetry group with unitary representations $\rho: G \to \mathcal{U}(\mathcal{H}_{1})$ on the physical space and $\pi: G \to \mathcal{U}(\mathcal{H}_{\frac{1}{2}})$ on the virtual space. The emission transition expectation $\mathcal{E}_{O,H}: \mathcal{B}(\mathcal{H}_{\frac{1}{2}}) \otimes \mathcal{B}(\mathcal{H}_{1}) \to \mathcal{B}(\mathcal{H}_{\frac{1}{2}})$ defined in \eqref{eq:EOH2} is $G$-equivariant. That is, for all $g \in G$,
\begin{equation}\label{eq:equivariance_correct}
    \mathcal{E}_{O,H}\left( \pi(g) X \pi(g)^\dagger \otimes \rho(g) Y \rho(g)^\dagger \right) = \pi(g) \, \mathcal{E}_{O,H}(X \otimes Y) \, \pi(g)^\dagger,
\end{equation}
for all $X \in \mathcal{B}(\mathcal{H}_{\frac{1}{2}})$ and $Y \in \mathcal{B}(\mathcal{H}_{1})$.
\end{theorem}
\begin{proof}
Let $g \in G$, $X \in \mathcal{B}(\mathcal{H}_{\frac{1}{2}})$, and $Y \in \mathcal{B}(\mathcal{H}_{1})$ be arbitrary. We compute the left-hand side of \eqref{eq:equivariance_correct}:
\[
\mathrm{LHS} = \mathcal{E}_{O,H}\left( \pi(g) X \pi(g)^\dagger \otimes \rho(g) Y \rho(g)^\dagger \right) = \sum_{m, m'} \langle m | \rho(g) Y \rho(g)^\dagger | m' \rangle \, A_m \left( \pi(g) X \pi(g)^\dagger \right) A_{m'}^\dagger.
\]
Expanding the matrix element in the physical basis $\{|m\rangle\}$ yields:
\[
\langle m | \rho(g) Y \rho(g)^\dagger | m' \rangle = \sum_{k, k'} \rho(g)_{m k} \langle k | Y | k' \rangle (\rho(g)^\dagger)_{k' m'}.
\]
Substituting this back into the expression for LHS and reordering the finite sums gives:
\[
\mathrm{LHS} = \sum_{k, k'} \langle k | Y | k' \rangle \left( \sum_{m, m'} \rho(g)_{m k} (\rho(g)^\dagger)_{k' m'} \, A_m \left( \pi(g) X \pi(g)^\dagger \right) A_{m'}^\dagger \right).
\]
Let $S_{k,k'}$ denote the inner sum. We can factor it as follows:
\[
S_{k,k'} = \left( \sum_{m} \rho(g)_{m k} A_m \right) \left( \pi(g) X \pi(g)^\dagger \right) \left( \sum_{m'} (\rho(g)^\dagger)_{k' m'} A_{m'}^\dagger \right).
\]
We now apply the covariance condition. The first factor is given directly by \eqref{eq:Cov}:
\[
\sum_{m} \rho(g)_{m k} A_m = \pi(g) A_k \pi(g)^\dagger.
\]
The third factor is the adjoint of the covariance condition for index $k'$:
\[
\sum_{m'} (\rho(g)^\dagger)_{k' m'} A_{m'}^\dagger = \left( \sum_{m'} \rho(g)_{m' k'} A_{m'} \right)^\dagger = \left( \pi(g) A_{k'} \pi(g)^\dagger \right)^\dagger = \pi(g) A_{k'}^\dagger \pi(g)^\dagger.
\]
Substituting these expressions into $S_{k,k'}$ and using the unitarity of $\pi(g)$ $(\pi(g)^\dagger \pi(g) = I)$ gives:
\[
S_{k,k'} = \left( \pi(g) A_k \pi(g)^\dagger \right) \left( \pi(g) X \pi(g)^\dagger \right) \left( \pi(g) A_{k'}^\dagger \pi(g)^\dagger \right) = \pi(g) A_k X A_{k'}^\dagger \pi(g)^\dagger.
\]
Substituting $S_{k,k'}$ back into the expression for LHS yields the final result:
\begin{align*}
\mathrm{LHS} &= \sum_{k, k'} \langle k | Y | k' \rangle \, \pi(g) A_k X A_{k'}^\dagger \pi(g)^\dagger \\
&= \pi(g) \left( \sum_{k, k'} \langle k | Y | k' \rangle A_k X A_{k'}^\dagger \right) \pi(g)^\dagger \\
&= \pi(g) \, \mathcal{E}_{O,H}(X \otimes Y) \, \pi(g)^\dagger.
\end{align*}
This completes the proof of \eqref{eq:equivariance_correct}.
\end{proof}

  Theorem \ref{thm:EOH-sym} provides a key understanding about the unique properties of AKLT state. It shows that the emission map is a perfect symmetry bridge. The equation (\ref{eq:equivariance_correct}) guarantees that the  symmetry of the hidden spin-½ particles  is perfectly converted into the standard symmetry of the physical spin-1 particles we observe.

 The dual of this map, which governs how information flows back to the hidden level, must also respect this symmetry. The  covariance implies that for a virtual density matrix \(\sigma\), the dual action satisfies:

\[
\mathcal{E}_{O,H}^*\left( \pi(g) \sigma \pi(g)^\dagger \right) = \left( \pi(g) \otimes \rho(g) \right) \mathcal{E}_{O,H}^*(\sigma) \left( \pi(g)^\dagger \otimes \rho(g)^\dagger \right).
\]

This means that any physical process acting on the spin-1 chain—like time evolution under a symmetric Hamiltonian—will automatically induce a corresponding action on the hidden spin-½  that faithfully preserves their projective character.

\section{  Conclusion }\label{Sect_Disc_Concl}

This work establishes a   HQMMs structure  on the AKLT spin chain. By identifying the virtual spin-½ chain  as the hidden quantum system, and the physical spin-1 sites as the observable output. The results transforms the algebraic-probabilistic  HQMMs  into a model grounded in many-body quantum physics, revealing  connections between quantum memory, topological order, and  correlation structures. The analysis of the hidden   chain demonstrates how the entangled structure of the AKLT state manifests as a non-classical memory resource. Furthermore, by extending the notion of  SPT  order   to the statistical properties of the emitted symbol sequences.

This framework naturally bridges to measurement-based quantum computation   \cite{BBDRN09} and quantum machine learning \cite{J22}.  Our HQMM construction, where transitions correspond to local operations and emissions to measurement outcomes, can be reinterpreted as a stochastic unraveling of an MBQC protocol. This perspective elevates the HQMM from a passive model of stochastic emissions to an active model of a distributed quantum computation unfolding in time. Consequently, questions about the expressivity and memory of the AKLT-HQMM translate directly into questions about the computational power and fault tolerance of MBQC on the AKLT resource state \cite{DKLP2002, Von08}.

Finally, this work charts several compelling research trajectories.  The extension to other SPT phases, symmetry groups, and spatial dimensions promises a vast landscape of HQMM behaviors to catalog and relate to topological invariants including the multi-dimensional case \cite{ikegami2025kitaevmeetsakltcompeting}.
One could investigate whether similar HQMM constructions can describe other SPT phases with several anomalous symmetries \cite{kap25}, or explore their applications in   machine learning \cite{SSP15}  where memory and symmetry play crucial roles~\cite{Li2024,Chen24}.

\section*{Declarations}

\subsection*{Conflict of Interest:}
The authors declare no conflict of interest.

\subsection*{Data Availability:}
No data were generated or analyzed during this study.

\subsection*{Ethics Approval:}
This work does not involve human participants or animals.

\subsection*{Funding:}
This research received no external funding.



\begin{thebibliography}{9}
\bibitem{Rab2002} Rabiner, Lawrence R. A tutorial on hidden Markov models and selected applications in speech recognition. Proceedings of the IEEE 77.2, 257-286 (2002).
\bibitem{Mor2021} Mor, Bhavya, Sunita Garhwal, and Ajay Kumar. A Systematic Review of Hidden Markov Models and Their Applications. Archives of computational methods in engineering 28, no. 3 (2021).

\bibitem{Li17} Liu, Y.  Investigation of Viterbi Algorithm Performance on Part-of-Speech Tagger of Natural Language Processing. In 2017 International Conference on Computer Systems, Electronics and Control (ICCSEC) (pp. 1430-1433). IEEE, (2017).

\bibitem{V03} Viterbi, A. Error bounds for convolutional codes and an asymptotically optimum decoding algorithm. IEEE transactions on Information Theory, 13(2), 260-269 (2003).
\bibitem{FV05} Forney, G. D. The viterbi algorithm. Proceedings of the IEEE, 61(3), 268-278,  (2005).


\bibitem{HY2018}  Hayashi, M., Yoshida, Y.: Asymptotic and non-asymptotic analysis for hidden Markovian process with  quantum hidden system. J. Phys. A Math. Theor. 51(33), 1 (2018)

\bibitem{Monras11}    A. Monras, A. Beige, and K. Wiesner, Hidden quantum Markov models and non-adaptive read-out of many-body states, Applied Mathematical and Computational Sciences, vol. 3, no. 1, pp. 93–122, (2011).

\bibitem{Srin2017} Srinivasan S, Gordon G, Boots B. Learning hidden quantum Markov models,  international conference on artificial intelligence and statistics. Lanzarote, Spain: PMLR, 09016   (2017).

\bibitem{MRDFS22} Markov, V., Rastunkov, V., Deshmukh, A., Fry, D.,  Stefanski, C.  Implementation and learning of quantum hidden Markov models. arXiv preprint arXiv:2212.03796 (2022).

\bibitem{Li2024} Li, X.-Y.; Zhu, Q.-S.; Hu, Y.; Wu, H.; Yang, G.-W.; Yu, L.-H.; Chen, G. A new quantum machine learning algorithm: Split hidden  quantum Markov model inspired by quantum conditional master equation. Quantum, 8, 1232  (2024).

\bibitem{FLW24} Fanizza, M., Lumbreras, J.,   Winter, A., Quantum theory in finite dimension cannot explain every general process with finite memory. Communications in Mathematical Physics, 405(2), 50  (2024).

\bibitem{AGLS24Q} Accardi L., Soueidy E.G., Lu Y.G., Souissi A.: Hidden Quantum Markov processes. Infin. Dimens. Anal. Quantum Probab. Relat. Top. (2024).

\bibitem{Abd23}A. Souissi, El Gheteb Soueidi, Entangled Hidden Markov Models,Chaos, Solitons  Fractals,174,113804,(2023).


 \bibitem{Sou25} Souissi, A. Matrix Product States as Observations of Entangled Hidden Markov Models: A. Souissi. Journal of Statistical Physics, 192(7), 88  (2025).

\bibitem{accardi1974}
Accardi, L.: Non-commutative Markov chains. In: Proceedings of the School of Mathematical Physics, University of Rome ``Tor Vergata''. pp. 268--295 (1974)


\bibitem{ASS20} Accardi, L., Souissi, A., \& Soueidy, E. G., Quantum Markov chains: A unification approach. Infinite Dimensional Analysis, Quantum Probability and Related Topics, 23(02), 2050016,  (2020).

\bibitem{AccOhy99}Accardi, L.,   Ohya, M. Compound channels, transition expectations, and liftings. Applied Mathematics and Optimization, 39(1), 33-59,  (1999).





\bibitem{FNW92} Fannes, M., Nachtergaele, B.,   Werner, R. F. Finitely correlated states on quantum spin chains. Communications in mathematical physics, 144(3), 443-490  (1992).


\bibitem{fannes2} Fannes M, Nachtergaele B and Werner R F 1992 Ground states of VBS models on Cayley trees J. Stat. Phys. 939–973

\bibitem{fannes3} Fannes M, Nachtergaele B and Werner R F 1992 Entropy estimates for finitely correlated states Annales de l'IHP Physique théorique 57(3) 259-277



\bibitem{ticozzi2008}
Ticozzi, F., Viola, L.: Quantum Markovian subsystems: invariance, attractivity, and control. IEEE Transactions on Automatic Control \textbf{53}(9), 2048--2063 (2008)



\bibitem{guan2024}
Guan, J., Feng, Y., Turrini, A., Ying, M.: Measurement-based verification of quantum markov chains. In: Computer Aided Verification. pp. 533--554. Springer Nature Switzerland, Cham (2024)

\bibitem{ying2016}
Ying, M.: Foundations of Quantum Programming. Morgan Kaufmann (2016)

\bibitem{wolf2012}
Wolf, M.M.: Quantum channels \& operations: guided tour. Lecture notes (2012). Available at https://mediatum.ub.tum.de/doc/1701036/document.pdf

\bibitem{ambainis2003}
Ambainis, A.: Quantum walks and their algorithmic applications. International Journal of Quantum Information \textbf{1}(4), 507--518 (2003)

\bibitem{AKLT1987} Affleck, I., Kennedy, T., Lieb, E. H., \& Tasaki, H. (1987). Rigorous results on valence-bond ground states in antiferromagnets. \textit{Physical Review Letters}, 59(7), 799.

\bibitem{Affleck1988} Affleck, I., Kennedy, T., Lieb, E. H., \& Tasaki, H. (1988). Valence bond ground states in isotropic quantum antiferromagnets, Communications in Mathematical Physics, 115(3), 477–528.


\bibitem{Haldine83} F. D. M. Haldane, Nonlinear field theory of large-spin heisenberg antiferromagnets: semiclassically quantized solitons of the
one-dimensional easy-axis néel state, Physical review letters 50, 1153 (1983).



\bibitem{Pollmann2010} F. Pollmann, A. M. Turner, E. Berg, and M. Oshikawa,  Entanglement spectrum of a topological phase in one dimension, Phys. Rev. B \textbf{81}, 064439 (2010).

\bibitem{Pollmann2012} ] F. Pollmann, E. Berg, A. M. Turner, and M. Oshikawa,
Symmetry protection of topological order in one-dimensional quantum spin systems, Phys. Rev. B 85,
075125 (2012).




\bibitem{Chen2011b} X. Chen, Z.-C. Gu, and X.-G. Wen,  Complete classification of one-dimensional gapped quantum phases in interacting spin systems, Phys. Rev. B \textbf{84}, 235128 (2011).

\bibitem{BR} O. Bratteli, D. W. Robinson, Operator Algebras and Quantum Statistical Mechanics 1 (Springer-Verlag, BerlinHeidelberg-New York, 1986).


\bibitem{O21} Ogata, Y.,  AZ 2-index of symmetry protected topological phases with reflection symmetry for quantum spin chains. Communications in Mathematical Physics, 385(3), 1245-1272 (2021).

\bibitem{O22} Ogata, Y., An invariant of symmetry protected topological phases with on-site finite group symmetry for two-dimensional Fermion systems. Communications in Mathematical Physics, 395(1), 405-457, (2022).

\bibitem{Tas23} H. Tasaki, Rigorous index theory for one-dimensional interacting topological insulators, J. Math. Phys. 64, 041903 (2023)


\bibitem{Liu2014} Liu, L. L.,   Hwang, T. Controlled remote state preparation protocols via AKLT states. \textit{Quantum Information Processing}, 13(7), 1639–1650 (2014). .

\bibitem{Schuch2011} N. Schuch, D. Pérez-García, and I. Cirac,  Classifying quantum phases using matrix product states and projected entangled pair states, Phys. Rev. B \textbf{84}, 165139 (2011).


\bibitem{CPSV21} Cirac, J. I., Perez-Garcia, D., Schuch, N.,   Verstraete, F., Matrix product states and projected entangled pair states: Concepts, symmetries, theorems. Reviews of Modern Physics, 93(4), 045003 (2021).

\bibitem{FC2024} Frembs, M.,  Cavalcanti, E. G.  Variations on the Choi–Jamiołkowski isomorphism. Journal of Physics A: Mathematical and Theoretical, 57(26), 265301 (2024).


\bibitem{Kitaev06} Kitaev, A.,  Preskill, J. Topological entanglement entropy. Phys. Rev. Lett. 96, 110404, DOI: 10.1103/PhysRevLett.96. 110404 (2006).


\bibitem{SV12}  Singh, S.,   Vidal, G., Tensor network states and algorithms in the presence of a global SU (2) symmetry. Physical Review B—Condensed Matter and Materials Physics, 86(19), 195114, (2012).

\bibitem{BCLY20} Brannan, M., Collins, B., Lee, H. H., Youn, S. G., Temperley–Lieb quantum channels. Communications in Mathematical Physics, 376(2), 795-839, (2020).




\bibitem{Chen2009} Chen, X., Zeng, B., Gu, Z.-C., Yoshida, B.,  Chuang, I. L. Gapped two-body Hamiltonian whose unique ground state is universal for one-way quantum computation. Physical Review Letters, 102(22), 220501  (2009).

\bibitem{Chen24} Chen, T., \& Byrnes, T. (2024). Efficient preparation of the AKLT state with measurement-based imaginary time evolution. Quantum, 8, 1557.


  \bibitem{Von08} M. Van den Nest and H. J. Briegel, Measurement-based quantum computation and undecidable logic, Found. Phys. 38(5) (2008), 448–457.


\bibitem{DKLP2002} E. Dennis, A. Kitaev, A. Landahl, and J. Preskill, Topological quantum memory, J. Math. Phys. 43, 4452 (2002).

\bibitem{ikegami2025kitaevmeetsakltcompeting}
Ikegami S, Fukui K, Pohle R, Motome Y. Kitaev Meets AKLT: Competing Quantum Disorder in Spin-3/2 Honeycomb Systems. arXiv:2512.06322   (2025).

\bibitem{BBDRN09} Briegel, H. J., Browne, D. E., Dür, W., Raussendorf, R.,   Van den Nest, M., Measurement-based quantum computation. Nature Physics, 5(1), 19-26, (2009).



\bibitem{Fidkowski2011} L. Fidkowski and A. Kitaev,  Topological phases of fermions in one dimension, Phys. Rev. B \textbf{83}, 075103 (2011).




\bibitem{J22} Jerbi, S., Fiderer, L. J., Poulsen Nautrup, H., Kübler, J. M., Briegel, H. J.,   Dunjko, V. Quantum machine learning beyond kernel methods. Nature Communications, 14(1), 517  (2023).








 \bibitem{SSP15} Schuld, M., Sinayskiy, I.,  Petruccione, F.  An introduction to quantum machine learning. Contemporary Physics, 56(2), 172-185 (2015).







\bibitem{SA25} Souissi, A.,   Andolsi, A.,  Dynamics of Matrix Product States in the Heisenberg Picture: Projectivity, Ergodicity, and Mixing. arXiv preprint arXiv:2503.06546, (2025).




 \bibitem{kap25} Kapustin, A.,  Sopenko, N.  Anomalous Symmetries of Quantum Spin Chains and a Generalization of the Lieb–Schultz–Mattis Theorem: A. Kapustin, N. Sopenko. Communications in Mathematical Physics, 406(10), 238 (2025).



\end{thebibliography}
\end{document}